\documentclass[journal]{IEEEtran}
\usepackage{graphicx}
\usepackage{epstopdf}
\usepackage{subfigure}
\usepackage{bm}
\usepackage{cite}
\usepackage{stfloats}
\usepackage{amssymb}
\usepackage{amsmath}
\usepackage{cases}
\usepackage{enumerate}
\usepackage{algorithm}
\usepackage{algorithmic}
\usepackage{color}
\usepackage{float}
\newtheorem{theorem}{Theorem}

\newtheorem{proof}{Proof}
\newtheorem{proposition}{Proposition}

\newcommand{\algorithmicinput}{\textbf{Input:}}
\newcommand{\INPUT}{\item[\algorithmicinput]}

\newcommand{\algorithmicoutput}{\textbf{Output:}}
\newcommand{\OUTPUT}{\item[\algorithmicoutput]}

\begin{document}

\title{Low-Complexity Linear Equalization for OTFS Systems with Rectangular Waveforms}

\author{Wenjun~Xu,~\IEEEmembership{Senior Member,~IEEE,} Tingting~Zou, Hui~Gao,~\IEEEmembership{Senior Member,~IEEE,} Zhisong~Bie, Zhiyong~Feng,~\IEEEmembership{Senior Member,~IEEE} and Zhiguo~Ding,~\IEEEmembership{Senior Member,~IEEE}   \thanks{W. Xu (Corresponding Author), T. Zou, H. Gao, Z. Bie and Z. Feng are with the Key Laboratory of Universal Wireless Communications, Ministry of Education, Beijing University of Posts and Telecommunications, Beijing 100876, China (Emails: {wjxu, zoutt, huigao, zhisongbie, fengzy}@bupt.edu.cn).

Z. Ding is with the School of Electrical and Electronic Engineering, the University of Manchester, UK (Email: zhiguo.ding@manchester.ac.uk).}}
\maketitle
\begin{abstract}
Orthogonal time frequency space (OTFS) is a promising technology for high-mobility wireless communications. However, the equalization realization in practical OTFS systems is a great challenge when the non-ideal rectangular waveforms are adopted. In this paper, first of all, we theoretically prove that the effective channel matrix under rectangular waveforms holds the block-circulant property and the special Fourier transform structure with time-domain channel. Then, by exploiting the proved property and structure, we further propose the corresponding low-complexity algorithms for two mainstream linear equalization methods, i.e., zero-forcing (ZF) and minimum mean square error (MMSE). Compared with the existing direct-matrix-inversion-based equalization, the complexities can be reduced by a few thousand times for ZF and a few hundred times for MMSE without any performance loss, when the numbers of symbols and subcarriers are both 32 in OTFS systems.
\end{abstract}

%
\IEEEpeerreviewmaketitle

\section{Introduction}
Orthogonal time frequency space (OTFS) is a kind of modulation technologies which can strongly combat the doppler effect in time-variant channels by spreading each data symbol on the delay-Doppler (DD) plane over the entire time-frequency (TF) plane~\cite{DBLP:journals/corr/abs-1802-02623}. Thus, OTFS is very robust in high-mobility scenarios, which is identified as a main challenge of future wireless communications.

Most existing work assumes that ideal waveforms are adopted in OTFS systems, which ensures the transmitter- and receiver-side waveforms satisfy bi-orthogonality~\cite{7925924,DBLP:journals/corr/abs-1808-00519,8786203}. This assumption can reap the advantages that the inter-symbol interference (ISI) and inter-carrier interference (ICI) are both perfectly eliminated, and constant channel gains are guaranteed for each data symbol in the DD domain~\cite{8424569}. However, the ideal waveform assumption is unavailable in practical scenarios, and the ISI and ICI is inevitably introduced, which causes non-uniform channel gains in the DD domain. As a consequence, channel matrix is no longer doubly circular~\cite{8424569}, which prohibitively increases the complexity of channel equalization, and hence hinders the realization of practical OTFS systems.

Therefore, low-complexity equalization plays a pivotal role in practical OTFS systems. Considering the potential of complexity reduction, linear equalization is usually much preferred, especially when the fractional doppler shift occurs, destroys the sparsity of the equivalent channel in the DD domain, and hence disables the feasibility of nonlinear equalization, e.g., message passing (MP)~\cite{7925924,DBLP:journals/corr/abs-1903-09441}. For the widely-used linear equalization, i.e., zero-forcing (ZF) and minimum mean square error (MMSE), the complexities of the existing direct-matrix-inversion-based schemes are both ${\cal O}( {{{\left( {NM} \right)}^3}} )$, where $M$, $N$ are the numbers of OTFS symbols and subcarriers, respectively. Typically, the value of $NM$ is in the order of thousand or even larger. Thus, even for the linear equalization of practical OTFS systems with non-ideal rectangular waveforms, the complexity is still overly intensive, and heavily limits the realization of practical OTFS systems.

To this end, we concentrate on low-complexity linear equalization algorithms for practical OTFS systems in this paper. Specifically, we first prove the block-circulant property and special Fourier transform structure of the effective channel matrix under rectangular waveforms. Then, low-complexity ZF and MMSE equalization algorithms are proposed by exploiting the proved property and structure. Our proposed algorithms significantly decrease the computational complexities from ${\cal O}( {{{\left( {NM} \right)}^3}} )$ to $\mathcal{O}\left( {{M^2}N\log N} \right)$ for ZF and ${\cal O}\left( {{M^2}{N^2}P} \right)$ for MMSE, where $P$ is the number of channel multi-paths.

In this paper, scalars, matrixes, and vectors are denoted by boldface capital, boldface lowercase, and normal letters, respectively. ${{\bf{I}}_N}$ is the identity matrix of size $N$. ${{\bf{A}}^{ - 1}}$, ${{\bf{A}}^{\rm{T}}}$, ${{\bf{A}}^{\rm{H}}}$ represent the inverse, transpose, and conjugate transpose of ${\bf{A}}$. The operation $\otimes $ represents Kronecker product. The operation ${\rm{circ}}\left\{ {{{\bf{a}}}} \right\}$ forms cyclic square matrix ${{\bf{A}}}$ whose first column is the vector ${{{\bf{a}}}}$, and ${\rm{circ}}\left\{ {{{\bf{A}}_1}, \cdot  \cdot {{\bf{A}}_N}} \right\}$ forms a block-circulant square matrix whose first column block as ${({\bf{A}}_{_1}^{\rm{T}}, \cdot  \cdot {\bf{A}}_{_N}^{\rm{T}})^{\rm{T}}}$. ${{\bf{F}}_M},{\bf{F}}_M^{\rm{H}} \in {\mathbb{C}^{M \times M}}$ are $M$-point DFT and IDFT matrices, respectively. ${[ {{{\bf{A}}}} ]_{i,p}}$ is the ${i_{{\rm{th}}}}$ row, ${p_{{\rm{th}}}}$ column element of ${{{\bf{A}}}}$, ${[ {{{\bf{a}}}} ]_{p}}$ is the ${p_{{\rm{th}}}}$ element of ${{{\bf{a}}}}$, and ${[ {{{\bf{A}}}} ]_{i}}$ is the ${i_{{\rm{th}}}}$ row vector of ${{{\bf{A}}}}$. $a:b$ denotes the indexes increasing from $a$ to $b$ one by one, and ${\left\langle  \cdot  \right\rangle _N}$ denotes the operation of modulus $N$.

\section{OFDM-Based OTFS Systems}
The realization of OTFS is compatible with OFDM by adding a pre-processing and post-processing block~\cite{7925924}, i.e., OFDM technologies can be adopted for multicarrier time-frequency signal transmission in OTFS systems~\cite{8119540,8422467,DBLP:journals/corr/abs-1903-09441}. In this section, a brief review of OFDM-based OTFS systems is given, and then the effective channel matrices of OTFS systems are derived for rectangular and ideal waveforms, respectively.

\subsection{Elements of OFDM-based OTFS Systems}
The time-frequency plane is defined as an $M \times N$ lattice which is sampled by interval $T$ (second) along the time axis and $\Delta f$ (Hz) along the frequency axis, respectively, i.e.,
\[\Lambda {\rm{ = }}\left\{ {\left( {m\Delta f,nT} \right):m = 1,...,M,n = 1,...,N} \right\}.\]
Delay-Doppler plane is defined as an \begin{small} $M \times N$ lattice \[\Gamma {\rm{ = }}\left\{ {\left( {\frac{l}{{M\Delta f}},\frac{k}{{NT}}} \right):l = 1,...,M,k = 1,...,N} \right\}.\vspace{-0.2cm}\]\end{small}%

As a pre-processing block added to OFDM systems at the transmitter, \emph{inverse symplectic finite Fourier transform} (ISFFT) converts the modulated data in DD domain to TF domain. ISFFT can be written in a matrix form as
\begin{equation}
{{\bf{X}}_{{\rm{TF}}}} = {\rm{ISFFT(}}{{\bf{X}}_{{\rm{DD}}}}{\rm{) = }}{{\bf{F}}_M}{{\bf{X}}_{{\rm{DD}}}}{\bf{F}}_N^{\rm{H}},
\end{equation}
where ${{\bf{X}}_{{\rm{TF}}}}$, \!${{\bf{X}}_{{\rm{DD}}}}\!\in \!{\mathbb{C}^{M \times N}}$ denote the discrete transmit data in a matrix form mapped on $\Lambda $ and $\Gamma $, respectively. At the receiver, the post-processing block additionally executes \emph{symplectic finite Fourier transform} (SFFT), which converts the output signal ${{\bf{R}}_{{\rm{TF}}}}$ on $\Lambda $ plane to DD-domain signal ${{\bf{Y}}_{{\rm{DD}}}}$,
\begin{equation}
{{\bf{Y}}_{{\rm{DD}}}} = {\rm{SFFT(}}{{\bf{R}}_{{\rm{TF}}}}{\rm{) = }}{\bf{F}}_M^{\rm{H}}{{\bf{R}}_{{\rm{TF}}}}{{\bf{F}}_N}.
\end{equation}
For the brevity of subsequent analysis, a DD-domain generalized input-output model is built for vectorized signal as
\begin{equation}\label{endtoend}
{{\bf{y}}_{{\rm{DD}}}} = {{\bf{H}}_{{\rm{e}}{\rm{f}}{\rm{f}}}}{{\bf{x}}_{{\rm{DD}}}}{\rm{ + }}{\bf{w}},
\end{equation}
where ${{\bf{x}}_{{\rm{DD}}}},{{\bf{y}}_{{\rm{DD}}}}\!\! \in\!\! {\mathbb{C}^{NM \times 1}}$ are obtained by column-wise stacking ${{\bf{X}}_{{\rm{DD}}}},{{\bf{Y}}_{{\rm{DD}}}}$ into column vectors, respectively. ${{\bf{H}}_{{\rm{eff}}}} \!\!\in\!\! {\mathbb{C}^{NM \times NM}}$ denotes the DD-domain end-to-end effective channel matrix and ${\bf{w}}$ is the noise vector. Note that we assume the same type of waveform at the transmitter and receiver. Details of ${{\bf{H}}_{{\rm{eff}}}}$ are discussed as follows.

\subsection{Effective Channel Matrix with Rectangular Waveforms}
In~\cite{8422467}, rectangular waveforms are considered in an OFDM-based OTFS system. The effective channel ${{\bf{H}}_{{\rm{eff}}}}$ can be expressed as
\begin{equation}\label{heff}
{{\bf{H}}_{{\rm{eff}}}}\!\! =\!\!(\underbrace {{{\bf{F}}_N} \!\otimes\! {\bf{F}}_M^{\rm{H}}}_{{\text{SFFT}}})\!{{{\bf{\Pi }}_{\rm{r}}}}\!(\!\underbrace{{{{\bf{G}}_{{\rm{r}}}}} \!\otimes \!{{\bf{F}}_M}}_{{{\text{OFDM Demod}}\text{.}}}\!){\bf{\tilde H}}(\underbrace{{{{{\bf{G}}_{{\rm{t}}}}} \!\otimes\! {\bf{F}}_M^{\rm{H}}}}_{{\text{OFDM Mod.}}}){{{\bf{\Pi }}_{\rm{t}}}}(\underbrace{{{\bf{F}}_N^{\rm{H}} \!\otimes \!{{\bf{F}}_M}}}_{\text{ISFFT}\!}),
\end{equation}
where ${{{\bf{\Pi }}_{\rm{r}}}},{{{\bf{\Pi }}_{\rm{t}}}} \!\!\in\!\! {\mathbb{C}^{MN \times MN}}$ and ${{\bf{G}}_{{\rm{r}}}},{{\bf{G}}_{{\rm{t}}}}\!\!\in\!\!{\mathbb{C}^{N \times N}}$ denote the receive, transmit windows and receive, transmit waveform matrices, respectively. Note that ${{{\bf{\Pi }}_{\rm{r}}}},{{{\bf{\Pi }}_{\rm{t}}}}\!\!= \!\!{{\bf{I}}_{MN}}$, and ${{\bf{G}}_{\rm{r}}},{{\bf{G}}_{\rm{t}}}\!\! =\!\! {{\bf{I}}_N}$ when the rectangular windows and waveforms are considered~\cite{8422467}. ${{\bf{\tilde H}}}\!\in \!{\mathbb{C}^{MN \times MN}}$ is the time-domain channel impulse response matrix of the vectorized time-domain signal after cyclic prefix (CP) adding and removing operations~\cite{8422467}. ${{\bf{\tilde H}}}$ is a block-diagonal matrix when the length of CP is larger than the maximum delay of the channel, and ${\bf{\tilde H}}$ is expressed as ${\bf{\tilde H}}\!\! = \!\!{\rm{diag}}({{\bf{\tilde H}}_1},{{\bf{\tilde H}}_2} \cdots {{\bf{\tilde H}}_N})$, where ${{\bf{\tilde H}}_p} \!\!\in\!\! {\mathbb{C}^{M \times M}}$. According to the tap-delay-line model~\cite{8119540}, there are only $P$ non-zero elements in each column and each row of ${{{\bf{\tilde H}}}_p}$ at fixed positions determined by the channel multi-path delay positions on the time axis, which is denoted as ${\bf{d}} \!\!=\! \!\left[ {{D_1},{D_2} \cdots {D_P}} \right]$. Therein, ${{D_P}}$ is the maximum delay position. Finally, for our considered system, the effective channel matrix in (\ref{heff}) is simplified to
\begin{equation}\label{recheff}
 {\bf{H}}_{_{{\rm{eff}}}}^{{\rm{rect}}} = \left( {{{\bf{F}}_N} \otimes {{\bf{I}}_M}} \right){\bf{\tilde H}}\left( {{\bf{F}}_N^{\rm{H}} \otimes {{\bf{I}}_M}} \right).
\end{equation}

\subsection{Channel Mismatch Under Ideal Waveform Assumption}
In practical wireless systems, if the effect of practical waveforms is not considered, i.e., assuming the ideal waveforms which cause that each data symbol experiences the constant channel gain at the receiver~\cite{8786203}, the systems may suffer the channel mismatch issues. We term this case as the ideal waveform assumption. In such a case, the channel gain of the first data symbol of ${{\bf{x}}_{{\rm{DD}}}}$, i.e., the first column of ${\bf{H}}_{{\rm{eff}}}^{{\rm{rect}}}$, is wrongly deemed as the constant channel gain, i.e., $[{\bf{H}}_{{\rm{eff}}}^{{\rm{rect}}}]_1^{\rm{T}}$. Furthermore, $[{\bf{H}}_{{\rm{eff}}}^{{\rm{rect}}}]_1^{\rm{T}}$ can be expressed as ${[ {{\bf{h}}_1^{\rm{T}},{\bf{h}}_2^{\rm{T}},...,{\bf{h}}_N^{\rm{T}}} ]^{\rm{T}}}$, wherein ${\bf{h}}_p \!\!\in \!\! {\mathbb{C}^{M \times 1}} $. Then, we can reshape $[{\bf{H}}_{{\rm{eff}}}^{{\rm{rect}}}]_1^{\rm{T}}$ into a matrix form as ${{\bf{H}}_{{\rm{DD}}}} \!\!=\!\![{{\bf{h}}_{\rm{1}}},{{\bf{h}}_{\rm{1}}},...,{{\bf{h}}_N}]\!\!\in\!\! {\mathbb{C}^{M \times N}}$. Hereby, each received data symbol can be rewritten by a two-dimension circular convolution between ${{\bf{H}}_{{\rm{DD}}}}$ and the transmit data as~\cite{DBLP:journals/corr/abs-1808-00519},
\begin{equation}\label{twocirc}
{\left[ {{{\bf{Y}}_{{\rm{DD}}}}} \right]_{k,l}}\!\! = \!\!\sum\limits_{k' = 1}^M  \sum\limits_{l' = 1}^N {{{\left[ {{{\bf{H}}_{{\rm{DD}}}}} \right]}_{k',l'}}{{\left[ {{{\bf{X}}_{{\rm{DD}}}}} \right]}_{{{\left\langle {k - k'} \right\rangle }_M} + 1,{{\left\langle {l - l'} \right\rangle }_N} + 1}}} .
\end{equation}
Based on (\ref{twocirc}), the mismatched effective channel matrix ${\bf{H}}_{{\rm{eff}}}^{{\rm{ide}}}$ under the ideal waveform assumption is given by
\begin{equation}\label{ideeff}
    {\bf{H}}_{{\rm{eff}}}^{{\rm{ide}}} \!=\! {\rm{circ}}\left\{ {{\rm{circ}}\left\{ {\bf{h}}_1 \right\},{\rm{circ}}\left\{ {\bf{h}}_2 \right\},\! \cdots \! ,{\rm{circ}}\left\{ {\bf{h}}_N \right\}} \right\}.
\end{equation}

\section{Low-Complexity ZF and MMSE Equalizations}
In this section, we first explore a deep insight into ${\bf{H}}_{{\rm{eff}}}^{{\rm{rect}}}$. Then, low-complexity ZF and MMSE equalization methods are proposed for practical OTFS systems with rectangular waveforms based on the structure feature of ${\bf{H}}_{{\rm{eff}}}^{{\rm{rect}}}$.

In order to facilitate the subsequent analysis, we first introduce two new operators. ${\rm{FF}}{{\rm{T}}_{{\rm{Mtx}}}}( \cdot )$ is defined as the Fourier transform between two sets of matrices. Taking ${{{\bf{\tilde H}}}_p},{{\bf{A}}_n} \!\in\! \mathbb{C}^{M \times M}$ as example, ${{{\bf{\tilde H}}}_p}{\rm{ = FF}}{{\rm{T}}_{{\rm{Mtx}}}}({{\bf{A}}_n})$ when the relationship ${{{\bf{\tilde H}}}_p} \!\!=\! \!\sum\nolimits_{n = 1}^N {{e^{-j\frac{{2\pi }}{N}\left( {p - 1} \right)\left( {n - 1} \right)}}}{{\bf{A}}_n}$ is satisfied. Similarly, we can define ${\rm{IFF}}{{\rm{T}}_{{\rm{Mtx}}}}( \cdot )$ by ${{\bf{A}}_n}{\rm{ = IFF}}{{\rm{T}}_{{\rm{Mtx}}}}({{{\bf{\tilde H}}}_p})$. Actually, when observing the fix-position element, e.g., $(m,l)$-element ${{[ {{{\bf{\tilde H}}}_p}]}_{m,l}}$ in the matrix ${{{\bf{\tilde H}}}_p}$, and re-organizing the fix-position elements, e.g., $(m,l)$-elements, from ${{{\bf{A}}_n}}$ as a vector by ${{\bf{a}}^{m,l}}\!\! = \!\!( {{{[{{\bf{A}}_1}]}_{m,l}},...,{{[{{\bf{A}}_N}]}_{m,l}}}) \!\!\in\!\! {\mathbb{C}^{N \times 1}}$, one can find ${{[ {{{\bf{\tilde H}}}_p}]}_{m,l}} =\!\! {\rm{FFT(}}{{{\bf{a}}^{m,l}}})$, where ${\rm{FFT(}}\cdot)$ is the traditional Fourier transform which can be implemented by fast Fourier transform (FFT). Since there are $M \times M$ different ${{\bf{a}}^{m,l}}$, ${{{\bf{\tilde H}}}_p}{\rm{ = FF}}{{\rm{T}}_{{\rm{Mtx}}}}({{\bf{A}}_n})$ is actually $M^2$ times Fourier transform operations. Thus, ${\rm{FF}}{{\rm{T}}_{{\rm{Mtx}}}}( \cdot )$, and ${\rm{IFF}}{{\rm{T}}_{{\rm{Mtx}}}}( \cdot )$ are both linear operators, which will be used for the following theorem.
\begin{theorem}\label{th1}
\emph{${\bf{H}}_{{\rm{eff}}}^{{\rm{rect}}}$ is a block-circulant matrix, i.e., ${{\bf{H}}_{{\rm{eff}}}^{{\rm{rect}}}}\!\! =\!\!{\rm{circ}}\{ {{\bf{A}}_1},{{\bf{A}}_2}\! \cdots,\!{{\bf{A}}_N}\} $, and the submatrices ${{\bf{A}}_n}$ have a Fourier transform relationship with channel impulse response ${{{\bf{\tilde H}}}_p}$ as ${{\bf{A}}_n}\! \!=\! \!{\rm{IFF}}{{\rm{T}}_{{\rm{Mtx}}}}({{{\bf{\tilde H}}}_p})$.}
\end{theorem}
\begin{proof}
${\bf{H}}_{{\rm{eff}}}^{{\rm{rect}}} \!\!\in\!\! {\mathbb{C}^{NM \times NM}}$ can be viewed as a block matrix, including $N \!\times\! N $ small square matrices denoted by ${\bf{H}}_{i,k}^{{\rm{rect}}}\! \in \!{\mathbb{C}^{M \!\times M}}$
, and from (\ref{recheff}), ${\bf{H}}_{i,k}^{{\rm{rect}}}$ is calculated as
\begin{equation}\label{heffgeneral}
{\bf{H}}_{i,k}^{{\rm{rect}}}= \sum\limits_{p = 1}^N {{f_{ip}}{f'_{pk}}} {{\bf{\tilde H}}_{p}}, i,k,p = 1,...,N,
\end{equation}
where ${f_{ip}} \!\!= \!\!{[ {{{\bf{F}}_N}} ]_{i,p}},{f'_{pk}}\!\! = \!\!{[ {{\bf{F}}_N^{\rm{H}}} ]_{p,k}}$. ${f_{ip}}{f'_{pk}}$ is further derived as
\begin{small}
\begin{equation}\label{fourpro}
\!\!{f_{ip}}{f'_{pk}}\!\! = \!\!{e^{ - j\frac{{2\pi }}{N}\left( {i - 1} \right)\left( {p - 1} \right)}}{e^{j\frac{{2\pi }}{N}\left( {p - 1} \right)\left( {k - 1} \right)}}\!\! = \!\!{e^{j\frac{{2\pi }}{N}\left( {p - 1} \right)\left( {k - i} \right)}}.
\end{equation}
\end{small}
The right hand side of (\ref{fourpro}) is only related to the difference of $i$ and $k$. By defining $n = {\left< {k\! - \! i} \right>_N} \!+ \!1$, and substituting $n$ into (\ref{fourpro}), we can obtain ${f_{ip}}{f'_{pk}}\!\! =\!\! {e^{j\frac{{2\pi }}{N}\left( {p - 1} \right)\left( {n - 1} \right)}}\!\!\triangleq\!\! { f_{n,p}}$, and then (\ref{heffgeneral}) can be simplified to
\begin{equation}\label{simgener}
{\bf{H}}_n^{{\rm{rect}}} = \sum\limits_{p = 1}^N {{f_{n,p}}} {{{\bf{\tilde H}}}_p} = {\rm{IFF}}{{\rm{T}}_{{\rm{Mtx}}}}({{{\bf{\tilde H}}}_p}).
\end{equation}
By defining ${{\bf{A}}_n} = {\bf{H}}_n^{{\rm{rect}}} \!\in\! {\mathbb{C}^{M \times M}}$, ${\bf{H}}_{{\rm{eff}}}^{{\rm{rect}}}$ can be rewritten in a block-circulant matrix form as
\begin{equation}\label{hcric}
{\bf{H}}_{{\rm{eff}}}^{{\rm{rect}}} \!\!= \!\!{\rm{circ}}\{ {{\bf{A}}_1},{{\bf{A}}_2} \cdots ,{{\bf{A}}_N}\},{{\bf{A}}_n} \!\!= \!\!{\rm{IFF}}{{\rm{T}}_{{\rm{Mtx}}}}({{{\bf{\tilde H}}}_p}).
\end{equation}
Then, the proof of Theorem~\ref{th1} is finished.
\end{proof}

As discussed before, under the ideal waveform assumption, the submatrices of the effective channel belong to circular matrices. In this case, the low-complexity ZF and MMSE schemes are more easily attained, since ${\bf{H}}_{{\rm{eff}}}^{{\rm{ide}}}$ can be diagonalized by FFT matrix thanks to the doubly circular property~\cite{8786203}.
Oppositely, with rectangular waveforms, the data symbol in DD domain does not experience a constant channel any more, which causes that ${{\bf{A}}_n}$ in~\eqref{hcric} is not a circulant matrix. The appealing diagonalization operation is not available for rectangular waveforms, which calls for new low-complexity methods based on Theorem~\ref{th1}.

\subsection{Low-Complexity ZF Equalization}
According to the vectorized model in (\ref{endtoend}), the vectorized output signal ${\bf{\hat x}}_{_{{\rm{DD}}}}^{\rm{ZF}}$ of ZF equalizer can be written by
\begin{equation}\label{zf}
{\bf{\hat x}}_{_{{\rm{DD}}}}^{{\rm{ZF}}} = {{\bf{W}}_{{\rm{ZF}}}}{{\bf{y}}_{{\rm{DD}}}} = {\bf{H}}_{{\rm{eff}}}^{ - 1}{{\bf{y}}_{{\rm{DD}}}},
\end{equation}
where ${{\bf{W}}_{{\rm{ZF}}}}$ denotes the evaluation matrix for ZF. Since direct inversion requires overwhelmingly high computational complexity, we develop a low-complexity algorithm. By referring to~\cite{1143132}, the inversion of a block-circulant matrix, i.e., ${\bf{H}}_{{\rm{eff}}}^{ - 1}$, is also a block-circulant matrix. Thus, we can represent ${{\bf{W}}_{{\rm{ZF}}}}$ as
\begin{equation}\label{hinv}
{{\bf{W}}_{{\rm{ZF}}}} = {\bf{H}}_{{\rm{eff}}}^{ - 1} = {\rm{circ}}\{ {{\bf{B}}_1},{{\bf{B}}_2} \cdots {{\bf{B}}_N}\},
\end{equation}
where ${{\bf{B}}_q} \!\!\in \!\!{\mathbb{C}^{M \times M}}$ can be obtained by ${{\bf{A}}_n}$ as~\cite{1143132}
\begin{align}\label{bp}
\nonumber {{\bf{B}}_q}\! = \!\frac{1}{N}\sum\limits_{t = 1}^N {{{\left( {{e^{j\frac{{2\pi (p - 1)}}{N}}}} \right)}^{N - t + 1}}{\bf{S}}_t^{ - 1}}& \! = \! \frac{1}{N}{\rm{FF}}{{\rm{T}}_{{\rm{Mtx}}}}({\bf{S}}_t^{ - 1}),\\
{{\bf{S}}_t} = \sum\limits_{n = 1}^N {{e^{j\frac{{\left( {t - 1} \right)\left( {k - 1} \right)}}{N}}}{{\bf{A}}_n}} & = {\rm{IFF}}{{\rm{T}}_{{\rm{Mtx}}}}({{\bf{A}}_n}),
\end{align}
with ${{\bf{S}}_t} \in {\mathbb{C}^{M \times M}}$. According to \eqref{hinv}, (\ref{bp}), to obtain ${{\bf{W}}_{{\rm{ZF}}}}$, the direct inversion of ${\bf{H}}_{{\rm{eff}}}^{{\rm{rect}}} \in {\mathbb{C}^{NM \times NM}}$ can be avoided by calculating the inversion of ${{\bf{S}}_t}$ and executing FFT operations. Moreover, it is observed that ${{\bf{S}}_t}$ has the same structure of ${{\bf{A}}_n}$. By exploiting the sparse property and structure knowledge of ${{\bf{S}}_t}$, we further develop a LU factorization method to obtain ${\bf{S}}_t^{ - 1}$ to proceed with the simplification of ZF equalization algorithm. The entire details of the proposed low-complexity ZF equalization are provided in Algorithm~\ref{alg1}.

\begin{algorithm}[tb]
\caption{Low-Complexity ZF Equalization}
\label{alg1}
\begin{algorithmic}[1]
\begin{small}
\INPUT  Delay position vector ${\bf{d}}$, multi-path number $P$, ${{{\bf{\tilde H}}}_p}$.\\
        \STATE  Define ${\bf{\Phi}} \!\!= \!\!{{\bf{I}}_M}$, vector ${\bf{u}}\!\! =\!\! [{\bf{d}},M \!- \!{D_P}]$, ${\bf{Y}}\!\! =\!\! {{\bf{0}}_M} \in {\mathbb{C}^{M \times M}}$.\\
        \STATE Compute submatrices ${{\bf{A}}_n}{\rm{ = IFF}}{{\rm{T}}_{{\rm{Mtx}}}}({{{\bf{\tilde H}}}_p})$.\\
        \STATE Compute ${{\bf{S}}_t}{\rm{ = IFF}}{{\rm{T}}_{{\rm{Mtx}}}}({{{\bf{A}}}_n})$.
\FOR {$t=1$ to $N$}
        \STATE  From $i \!= \!1$ to $M \!- \!{D_P}$, ${\left[ {\bf{\Phi }} \right]_{i + {\bf{u}},i}} = {\textstyle{{{{\left[ {{{\bf{S}}_t}} \right]}_{i + {\bf{u}},i}}} \over {{{\left[ {{{\bf{S}}_t}} \right]}_{i,i}}}}}$
\FOR{$n = 1$ to $P+1$}
      \STATE $i = {[{\bf{u}}]_n} + 1:{[{\bf{u}}]_{n + 1}}$;
       \STATE  ${\left[ {{\bf{LU}}} \right]_i} = {\left[ {{{\bf{S}}_t}} \right]_i} - \sum\limits_{k = 2}^n {\sum\limits_{j = i}^M {{{\left[ {{\bf{\Phi}}} \right]}_{i,i - {D_k}}}{{\left[ {{\bf{LU}}} \right]}_{i - {D_k},j}}} } $.
     \ENDFOR
       \STATE Update the remaining part of ${\bf{\Phi}}$ by regular LU factorization.
\FOR{$n = 1$ to $P$}
      \STATE $k = {[{\bf{u}}]_n} + 1:{[{\bf{u}}]_{n + 1}}$;
      \STATE ${\left[ {\bf{Y}} \right]_k} = {\left[ {{{\bf{I}}_M}} \right]_k} - \sum\limits_{i = 1}^n {{{\left[ {{\bf{\Phi}}} \right]}_{k,k - {\bf{u}}(i)}}{{\left[ {\bf{Y}} \right]}_{k - {\bf{u}}(i)}}} $.
      \ENDFOR
      \STATE From $k \!\!=\!\! {D_P} \!+\! 1$ to $M$, ${\left[ {\bf{Y}} \right]_k} = {\left[ {{{\bf{I}}_M}} \right]_k} - {\left[ {{\bf{\Phi}}} \right]_{k,1:k - 1}}{\left[ {\bf{Y}} \right]_{1:k - 1}}$.
      \STATE From $k \!\!=\!\! N$ to $1$, set $f$ as the maximum between $M - {D_{P}}$ and ${k + 1}$, then ${\left[ {{\bf{S}}_t^{ - 1}} \right]_k} = {\left[ {\bf{Y}} \right]_k} - {\left[ {\bf{\Phi }} \right]_{k,f:M}}{\left[ {\bf{Y}} \right]_{f:M}}$.
    \ENDFOR
    \STATE Compute ${{\bf{B}}_q}{\rm{ = }}\frac{1}{N}{\rm{FF}}{{\rm{T}}_{{\rm{Mtx}}}}({\bf{S}}_{_t}^{ - 1})$.

    \OUTPUT Evaluation matrix ${{\bf{W}}_{{\rm{ZF}}}}\! = \!{\rm{circ}}\{ {{\bf{B}}_1},{{\bf{B}}_2} \cdots {{\bf{B}}_N}\}$.
\end{small}
\end{algorithmic}
\end{algorithm}

\subsection{Low-Complexity MMSE Equalization}
ZF may cause the noise enhancement issue, which can be addressed by the MMSE equalization given by
\begin{equation}\label{mmse}
{\bf{\hat x}}_{_{{\rm{DD}}}}^{{\rm{MMSE}}}\!\! = \!\!{{\bf{W}}_{{\rm{MMSE}}}}{{\bf{y}}_{{\rm{DD}}}} \!\! = \!\! {\left( {{\bf{H}}_{{\rm{eff}}}^{\rm{H}}{{\bf{H}}_{{\rm{eff}}}} + {\sigma ^2}{{\bf{I}}_{MN}}} \right)^{ - 1}}{\bf{H}}_{{\rm{eff}}}^{\rm{H}}{{\bf{y}}_{{\rm{DD}}}},
\end{equation}
where ${{\bf{W}}_{{\rm{MMSE}}}}$ denotes the evaluation matrix for MMSE, and ${{\sigma ^2}}$ is the noise energy.
\begin{proposition}\label{th2}
The evaluation matrix ${{\bf{W}}_{{\rm{MMSE}}}}$ is a block-circulant matrix and the submatrices can be expressed as
\begin{equation}\label{mmseblock}
\begin{small}
{{\bf{W}}_{{\rm{MMSE}}}} = {\rm{circ}}\left\{ {{{{\bf{\tilde W}}}_1},{{{\bf{\tilde W}}}_2} ,\cdots ,{{{\bf{\tilde W}}}_N}} \right\},
\end{small}
\end{equation}
where ${{{\bf{\tilde W}}}_k} \!\!=\!\! \sum\nolimits_{i = 1}^N {{\bf{C}}_i^{ - 1}{\bf{A}}_{{{\left[ {i -k} \right]}_N}\! +\! 1}^{\rm{H}}}\! \!\! \in\! \!\!{\mathbb{C}^{M \times M}}$, and ${\bf{C}}_i^{ - 1}\!\! \in\!\! {\mathbb{C}^{M \times M}}$ is the submatrices of the block-circulant matrix ${{\bf{C}}^{ - 1}} = {\left( {{\bf{H}}_{{\rm{eff}}}^{\rm{H}}{{\bf{H}}_{{\rm{eff}}}} + {\sigma ^2}{\bf{I}}} \right)^{ - 1}} = {\rm{cric}}\{ {\bf{C}}_1^{ - 1},{\bf{C}}_2^{ - 1}, \cdots ,{\bf{C}}_N^{ - 1}\}$.
\end{proposition}
\begin{proof}
The conjugate transpose of effective channel can be expressed as ${\bf{H}}_{{\rm{eff}}}^{\rm{H}}\!\! = \!\!{\rm{circ}}\{ {\bf{A}}_1^{\rm{H}},{\bf{A}}_N^{\rm{H}},{\bf{A}}_{N - 1}^{\rm{H}}, \cdots ,{\bf{A}}_2^{\rm{H}}\}$, and hence ${\bf{C}}\!\! =\!\! \left( {{\bf{H}}_{{\rm{eff}}}^{\rm{H}}{{\bf{H}}_{{\rm{eff}}}} + {\sigma ^2}{\bf{I}}} \right)$ is the product of two block-circulant matrices plus a diagonal matrix. The product of two block-circulant matrices is also a block-circulant matrix, which means both ${\bf{C}}$ and ${{\bf{C}}^{ - 1}}$ are block-circulant matrices. Thus, ${{\bf{W}}_{{\rm{MMSE}}}}$ is a block-circulant matrix, and a general expression of submatrices in ${{\bf{W}}_{{\rm{MMSE}}}}$ can be obtained by using block-circulant matrix multiplication twice, as shown in~(\ref{mmseblock}).
\end{proof}

Rewrite ${\bf{C}}$ as ${\bf{C}} \!\!=\!\! {\rm{cric}}\{ {\bf{C}}_1^{},{\bf{C}}_2^{}, \cdots ,{\bf{C}}_N^{}\}$. According to (\ref{mmseblock}), ${{\bf{C}}^{{\rm{ - 1}}}} \!\!\in\!\! {\mathbb{C}^{NM \times NM}}$ is required when obtaining ${{\bf{W}}_{{\rm{MMSE}}}}$. We can apply (\ref{bp}) to simplify the inversion operation as
\begin{equation}\label{mmseinv}
{\bf{C}}_q^{ - 1} = \frac{1}{N}{\rm{FF}}{{\rm{T}}_{{\rm{Mtx}}}}({\bf{S}}_t^{ - 1}),{{\bf{S}}_t} = {\rm{IFF}}{{\rm{T}}_{{\rm{Mtx}}}}({{\bf{C}}_n}).
\end{equation}
Based on the analysis above, the direct inversion of MMSE equalization can also be avoided by using a similar algorithm as Algorithm~\ref{alg1}.

\subsection{Complexity Analysis}
As there are $M^2$ times $N$-point FFT (IFFT) operations, the computational complexity of ${\rm{FF}}{{\rm{T}}_{{\rm{Mtx}}}}( \cdot )$ (${\rm{IFF}}{{\rm{T}}_{{\rm{Mtx}}}}( \cdot )$) is $\mathcal{O}\left( {{M^2}N\log N} \right)$. Moreover, the LU factorization method reduces the complexity of direct inversion of ${\bf{S}}_t^{}$ from ${\cal O}\left( {{M^3}} \right)$ to ${\cal O}\left( {{M^2}{D_P}} \right)$, usually with ${D_P} \ll M$. As for the multiplication of the matrices in (\ref{mmseblock}), by exploiting the sparsity of the matrices, the complexity is reduced from ${\cal O}( {{{\left( {NM} \right)}^3}} )$ to ${\cal O}\left( {{M^2}{N^2}P} \right)$. Finally, Table~\uppercase\expandafter{\romannumeral1} shows the computational complexity comparison between the direct-matrix-inversion-based equalization and the proposed low-complexity equalization. The effective channel matrix in the DD domain is obtained by (\ref{simgener}), which is also simplified via the operation ${\rm{IFF}}{{\rm{T}}_{{\rm{Mtx}}}}( \cdot )$. For $N=M=32,P=6$, the overall complexity can be decreased by a few thousand times for ZF and a few hundred times for MMSE. More complexity reduction can be achieved for practical systems with larger $N,M,P$.

\section{Simulation Results}
In this section, simulation results are provided. The parameters are set as follows: The carrier frequency is 5~GHz, the subcarrier space $\Delta f$ is 15~KHz, the number of subcarrier $M$ is $64$, the number of OTFS symbols $N$ in a transmission block is $32$, and the channel is modeled as Vehicular Channel B~\cite{ahmadi2007channel}. Hereby, ${D_P}$ is 20 and let ${L_{{\rm{CP}}}}\!\! =\!\! {D_P}\! +\! 1$. The 4-QAM modulated signal is adopted and the maximum doppler spread is ${f_{\max }} \!\!=\!\! 1$~KHz. The simulation results are averaged over 20000 times realizations.
\begin{figure}[!htbp]
  \centering
    \vspace{-0.3cm}
  \includegraphics[width=3in]{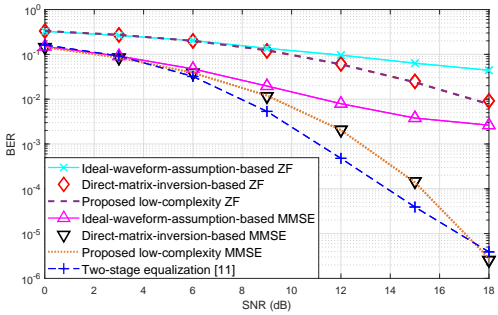}\\
  \setlength{\abovecaptionskip}{0pt}
\setlength{\belowcaptionskip}{-2cm}
  \caption{BER performance comparison of different schemes.}\label{fig1}
    \vspace{-0.3cm}
\end{figure}

Fig.~\ref{fig1} indicates that our proposed low-complexity equalization algorithms result in no performance loss compared with the direct-inversion-based schemes~\cite{8422467}. The significant performance degradation is observed by the inappropriate ideal-waveform assumption~\cite{8786203}, due to that the interference under rectangular waveforms gives rise to the variation of channel gains among different data symbols in the DD domain, which causes a mismatch between the theoretical and practical channels. Note that a two-stage equalization scheme, resembling traditional feedback decision equalization, has been recently proposed in~\cite{DBLP:journals/corr/abs-1709-02505}. From Fig.~\ref{fig1} and Table~\uppercase\expandafter{\romannumeral1} ($\Gamma$ is the complexity induced by the hard-decision in the second stage), the complexity of our proposed schemes is at least a few hundreds times lower than that of~\cite{DBLP:journals/corr/abs-1709-02505} with slight performance loss during the signal-noise-ratio (SNR) range of 5-16~dB. Besides, more advantages of our schemes are as follows: The linear equalization only requires the easily-obtained DD domain channel~\cite{8422467}, while the two-stage equalization requires the intractable frequency-domain response of time-varying channel in practical OTFS systems; the two-stage equalization imposes higher hardware requirements on practical systems.

\begin{table}[!htbp]
  \centering
  \vspace{-0.3cm}
  \caption{Complexities of Different Schemes}\label{time}
  \vspace{-0.1cm}
  \begin{tabular}{|c|c|}
  \hline
  \bfseries  Scheme  &\bfseries Complexity \\
  \hline
  Direct-inversion-based ZF/MMSE~\cite{8422467}  &  $\mathcal{O}( {{{( {NM} )}^3}} )$  \\
  \hline
  Two-stage equalization in~\cite{DBLP:journals/corr/abs-1709-02505}  & $\mathcal{O}( {{{( {NM} )}^3}} ){\rm{ + }}\Gamma $\\
  \hline
  Proposed low-complexity ZF/MMSE    & ${\cal O}\!( {{M^2}N\!\log \!N } \!)$/${\cal O}( {{M^2}{N^2}P} )$ \\
  \hline
  \end{tabular}
  \vspace{-0.5cm}
\end{table}

\section{Conclusion}
In this paper, we identify and prove the effective channel matrix structure of OTFS systems with rectangular waveforms. Based on the proved structure feature, new low-complexity linear equalization methods are proposed for ZF and MMSE, respectively. Analysis and simulation results validate that the proposed linear equalization methods can significantly reduce the computational complexity with no performance loss, which enables the future applications of OTFS technologies to practical communication systems.


\ifCLASSOPTIONcaptionsoff
  \newpage
\fi

\bibliographystyle{IEEEtran}
\bibliography{IEEEabrv,Ref_19_COMLET_0914}
\end{document}